\documentclass[final]{nextgame2026}

\usepackage[T1]{fontenc}
\usepackage{mathtools}
\usepackage{booktabs}
\usepackage{enumitem}
\usepackage{url}

\title[Optimism as a Vulnerability]{Optimism as a Vulnerability: Deceptive Stackelberg Control of UCB Bandit Followers}

\optauthor{%
\Name{{\c{S}}uayp Talha Kocabay} \Email{kocabaysuayptalha08@gmail.com}\\
\addr Independent Researcher
\AND
\Name{Kerem Yal\c{c}\i{n}} \Email{54keremyalcin@gmail.com}\\
\addr Independent Researcher
\AND
\Name{Talha R\"uzgar Akku\c{s}} \Email{talharuzgarakkus@gmail.com}\\
\addr Independent Researcher}

\newtheorem{assumption}[theorem]{Assumption}

\newcommand{\AL}{A_L}
\newcommand{\AF}{A_F}
\newcommand{\DeltaL}{\Delta(\AL)}
\newcommand{\UL}{U_L}
\newcommand{\UF}{U_F}
\newcommand{\br}{\mathrm{BR}}
\newcommand{\sse}{\mathrm{SSE}}

\newcommand{\one}{\mathbf{1}}

\begin{document}

\maketitle

\begin{abstract}
Upper Confidence Bound (UCB) algorithms guarantee sublinear regret for agents learning unknown stochastic environments, yet the same principle that makes them statistically efficient---optimism in the face of uncertainty---induces a predictable strategic vulnerability against an omniscient adaptive leader. Classical strong Stackelberg equilibrium (SSE) assumes that the follower immediately best-responds to the leader's committed mixed action; it therefore supplies no mechanism-design prescription for a leader facing a boundedly rational follower who constructs and acts on empirical reward histories. We formalize this conflict in a finite-horizon repeated Stackelberg game and give exact constructive proofs for a deceptive leader mechanism. In a honeypot phase, the leader pays a finite signaling cost to inflate the UCB index of a designated follower action. In a trap phase, the leader switches to a selfish action distribution while the follower remains locked into the designated action because the manipulated empirical history and exploration bonus dominate competing indices. Under explicit separation and payoff assumptions, the leader's cumulative utility strictly exceeds the classical SSE ceiling, and the manipulation cost is bounded by a regret calculation of order $O(\sqrt{T\ln T})$. The results identify a formal incompatibility between static equilibrium prescriptions and dynamically learned empirical incentives.
\end{abstract}

\begin{keywords}
Stackelberg games, bandit learning, UCB, strategic deception, reward manipulation
\end{keywords}

\section{Introduction}

Stackelberg security games and related leader--follower models typically analyze commitment to a mixed action followed by a rational best response \citep{von2010leadership,korzhyk2011stackelberg,conitzer2006computing,tambe2011security}. In the strong Stackelberg equilibrium convention, the follower breaks ties in favor of the leader; the leader consequently solves a static optimization problem over induced best responses. This model is internally coherent when the follower observes payoffs and responds as a utility maximizer. It is not a model of a follower who learns payoffs from interaction.

Bandit-learning followers instantiate a different behavioral primitive. UCB1 \citep{auer2002finite,lattimore2020bandit,slivkins2019introduction} is no-regret in stationary stochastic bandits: arms with uncertain value receive an optimism bonus, and suboptimal arms are sampled only logarithmically often. In a strategic environment, however, reward samples are not exogenous evidence about a fixed arm. They are data produced by another player. This places the model closer to learning in games and nonstationary multi-agent learning \citep{fudenberg1998theory,cesa2006prediction,shoham2007multiagent,nowe2012game} than to one-shot commitment. An adaptive leader can therefore treat the follower's statistical estimator as an object of control.

The closest technical literature is reward or action poisoning of bandit learners, where an external attacker corrupts feedback or actions to force target pulls at small perturbation cost \citep{jun2018adversarial,liu2019data,liu2020action,ma2023adversarial,balasubramanian2024cmab,wang2024stealthy,zuo2023ucb}. Our mechanism differs because the leader does not edit rewards exogenously; it creates the reward stream endogenously through legal Stackelberg play. Verification and corruption-aware bandits \citep{rangi2021secure,lykouris2018stochastic,xu2024robustts,kim2026heteroskedastic} and Stackelberg learning with manipulative or non-myopic agents \citep{haghtalab2022nonmyopic,yu2024decentralized,birmpas2021deceiving,nguyen2019imitative} motivate the defensive discussion below.

This paper makes the preceding claim precise. We compare two leader models in a repeated finite Stackelberg game. The baseline leader commits myopically to a classical SSE action and assumes immediate best response. The deceptive leader instead first rewards a target follower action $j^\star$ to raise its empirical mean, then switches to a leader-favorable action distribution for which $j^\star$ is no longer follower-optimal. The follower nevertheless continues selecting $j^\star$ for a calculable number of rounds because its UCB index remains larger than all competitors.

Our contribution is theoretical. We do not claim that every Stackelberg game admits profitable deception. We identify a verifiable class of games satisfying payoff separation, targetability, and exploitability assumptions, and prove that in this class the deceptive leader obtains strictly more cumulative utility than the static SSE ceiling. The proof exposes the failure mode: static mechanism design optimizes against the best-response correspondence $y\in\arg\max_{a\in\AF}\UF(x,a)$, whereas an empirical follower implements a stateful map from histories to actions.

\section{Model}

Let $\AL=\{1,\ldots,m\}$ and $\AF=\{1,\ldots,n\}$ be finite action sets. A $T$-step repeated Stackelberg game is specified by mean utilities
\[
  \UL,\UF:\AL\times\AF\to[0,1].
\]
At each time $t$, the leader selects $x_t\in\DeltaL$, the follower selects $y_t\in\AF$, and the realized follower reward is
\[
  R^F_t=\UF(x_t,y_t)+\eta_t,
\]
where $(\eta_t)$ is conditionally mean-zero and $\sigma$-sub-Gaussian; the deterministic model is $\sigma=0$. The leader's payoff is its expected utility $\UL(x_t,y_t)=\sum_i x_t(i)\UL(i,y_t)$. The main theorem assumes the leader observes the follower's sufficient statistics $(N_j,\hat\mu_j)_j$ and knows both payoff matrices; Appendix~\ref{app:partial} records what changes under action-only observation. The follower observes only its realized action and scalar reward.

\begin{definition}[Classical SSE value]
For a mixed leader action $x$, define the follower best-response set
\[
  \br(x)=\arg\max_{j\in\AF}\UF(x,j).
\]
The strong Stackelberg value is
\[
  V_{\sse}
  =
  \max_{x\in\DeltaL}\max_{j\in\br(x)} \UL(x,j).
\]
\end{definition}

The SSE abstraction imposes $y_t\in\br(x_t)$ at every round. This converts the interaction into a static commitment problem. The follower studied here instead uses UCB1. Let $N_j(t)=\sum_{s\le t}\one\{y_s=j\}$ and let $\hat\mu_j(t)$ be the empirical mean of rewards observed by the follower when it played $j$ up to time $t$. Initially every arm is pulled once, or equivalently $N_j(0)=0$ and the algorithm uses an initialization schedule over $n$ rounds. Thereafter, for exploration parameter $c>0$,
\[
  y_t\in
  \arg\max_{j\in\AF}
  \left[
    \hat\mu_j(t-1)
    +c\sqrt{\frac{\ln t}{N_j(t-1)}}
  \right],
\]
where unplayed arms have index $+\infty$.

The crucial difference is that $\br(x_t)$ is a function of the current leader mixture, whereas UCB is a function of the entire empirical history. Thus, a leader can manipulate the sufficient statistics $(\hat\mu_j,N_j)_{j\in\AF}$ before changing $x_t$.

\section{Deceptive Leader Mechanism}

Fix a target follower action $j^\star\in\AF$. The mechanism has two phases. During the honeypot phase the leader selects mixtures that give high follower reward to $j^\star$. During the trap phase the leader selects a selfish exploitative mixture $x^{\mathrm{exp}}$ that maximizes the leader's payoff conditional on the follower choosing $j^\star$.

\begin{assumption}[Targetability]
There exists a honeypot mixture $x^{\mathrm{hon}}\in\DeltaL$ and numbers $\alpha,\beta\in[0,1]$ with $\alpha>\beta$ such that
\[
  \UF(x^{\mathrm{hon}},j^\star)=\alpha,
  \qquad
  \UF(x^{\mathrm{hon}},j)\le \beta
  \quad\forall j\ne j^\star.
\]
\end{assumption}

\begin{assumption}[Exploitability and follower harm]
There exists $x^{\mathrm{exp}}\in\DeltaL$ and a constant $\rho\in[0,1]$ such that
\[
  \UF(x^{\mathrm{exp}},j^\star)=\rho,\qquad
  \max_{j\ne j^\star}\UF(x^{\mathrm{exp}},j)=\rho+\gamma_F
\]
for some follower suboptimality gap $\gamma_F>0$,
and
\[
  L^\star:=\UL(x^{\mathrm{exp}},j^\star)>V_{\sse}.
\]
\end{assumption}

The second condition is the source of profit and deception: $j^\star$ is strictly suboptimal for the follower under $x^{\mathrm{exp}}$, yet unusually valuable for the leader.

\begin{algorithm2e}
\caption{Deceptive Leader Mechanism}
\KwData{horizon $T$, target $j^\star$, UCB parameter $c$, mixtures $x^{\mathrm{hon}},x^{\mathrm{exp}}$}
Initialize the follower by allowing one sample of each action\;
\For{$t=n+1,\ldots,n+\tau_1$}{
  play $x_t=x^{\mathrm{hon}}$\;
}
\For{$t=n+\tau_1+1,\ldots,T$}{
  play $x_t=x^{\mathrm{exp}}$\;
}
\end{algorithm2e}

\section{Honeypot Inflation}

Let $t_0=n$ denote the end of forced initialization. Suppose every non-target action has bounded initialization history and $j^\star$ has initial count $N_0\ge1$ and empirical mean $\hat\mu_0$. The honeypot phase repeatedly plays $x^{\mathrm{hon}}$ until $j^\star$ has $\tau_1$ additional observations of reward $\alpha$.

\begin{lemma}[Exact target mean after honeypot]
If $j^\star$ has $N_0\ge 1$ initialization samples with empirical mean $\hat\mu_0$, then after $\tau_1$ honeypot samples,
\[
  \hat\mu_{j^\star}(t_0+\tau_1)
  =
  \frac{N_0\hat\mu_0+\tau_1\alpha}{N_0+\tau_1}.
\]
Consequently, for any $\varepsilon>0$,
\[
  \hat\mu_{j^\star}(t_0+\tau_1)\ge \alpha-\varepsilon
  \quad\text{whenever}\quad
  \tau_1\ge N_0\frac{\alpha-\hat\mu_0-\varepsilon}{\varepsilon}.
\]
\end{lemma}

\begin{proof}
Appendix~\ref{app:proofs} gives the algebra.
\end{proof}

\begin{theorem}[Minimum honeypot length for strict UCB dominance]
Fix a desired post-honeypot time $s=t_0+\tau_1+1$. Suppose each non-target action $j\ne j^\star$ has $1\le N_j(t_0)\le B_N$ and empirical mean at most $B_\mu$. After $\tau_1$ honeypot samples of $j^\star$, the target UCB index strictly dominates every non-target index at time $s$ if and only if
\[
  \frac{N_0\hat\mu_0+\tau_1\alpha}{N_0+\tau_1}
  +c\sqrt{\frac{\ln s}{N_0+\tau_1}}
  >
  B_\mu+c\sqrt{\ln s}\max_{j\ne j^\star}\frac{1}{\sqrt{N_j(t_0)}}.
\]
In particular, a sufficient explicit condition is
\[
  \tau_1
  \ge
  \min\left\{q\in\mathbb{N}:
  \frac{N_0\hat\mu_0+q\alpha}{N_0+q}
  +c\sqrt{\frac{\ln(t_0+q+1)}{N_0+q}}
  >
  B_\mu+c\sqrt{\ln(t_0+q+1)}
  \right\},
\]
where the right-hand side uses the conservative bound $N_j(t_0)\ge1$.
\end{theorem}

\begin{proof}
Appendix~\ref{app:proofs} expands the index comparison term by term.
\end{proof}

For stochastic rewards, Appendix~\ref{app:stochastic} gives a high-probability analogue obtained by subtracting uniform sub-Gaussian confidence radii from the target index and adding them to the competing indices.

\section{Trap Duration}

Let $\tau=t_0+\tau_1$ be the switch time. During exploitation the follower receives reward $\rho$ whenever it plays $j^\star$. Let $N_\star=N_0+\tau_1$ and let
\[
  M_\star=N_0\hat\mu_0+\tau_1\alpha.
\]
If the follower keeps selecting $j^\star$ for $q$ exploitation rounds, then at time $\tau+q+1$ its target empirical mean is
\[
  \bar\mu_\star(q)=\frac{M_\star+q\rho}{N_\star+q}.
\]
The corresponding target index is
\[
  I_\star(q)=
  \bar\mu_\star(q)+c\sqrt{\frac{\ln(\tau+q+1)}{N_\star+q}}.
\]
For a non-target action $j$, define the frozen comparator index
\[
  J_j(q)=\hat\mu_j(t_0)+c\sqrt{\frac{\ln(\tau+q+1)}{N_j(t_0)}}.
\]

\begin{theorem}[Exact lock-in duration]
Assume deterministic rewards. Conditional on strict dominance at the switch, the exact number of consecutive exploitation rounds during which UCB selects $j^\star$ is
\[
  \Delta
  =
  \max\left\{
    q\in\{0,\ldots,T-\tau\}:
    I_\star(r)> \max_{j\ne j^\star}J_j(r)
    \text{ for all } r=0,\ldots,q-1
  \right\}.
\]
Equivalently, the first escape time is
\[
  q_{\mathrm{esc}}
  =
  \min\left\{
    q\ge0:
    \frac{M_\star+q\rho}{N_\star+q}
    +c\sqrt{\frac{\ln(\tau+q+1)}{N_\star+q}}
    \le
    \max_{j\ne j^\star}
    \left[
      \hat\mu_j(t_0)+c\sqrt{\frac{\ln(\tau+q+1)}{N_j(t_0)}}
    \right]
  \right\},
\]
with $\Delta=\min\{q_{\mathrm{esc}},T-\tau\}$ when the minimum exists and $\Delta=T-\tau$ otherwise.
\end{theorem}

\begin{proof}
Appendix~\ref{app:proofs} proves the statement by induction on exploitation rounds.
\end{proof}

Appendix~\ref{app:stochastic} also gives a high-probability lower bound on $\Delta$ by requiring the confidence-adjusted target index to dominate the frozen competitors for every $q<Q$.

\section{Utility Above the SSE Ceiling}

Let $H^\star=\UL(x^{\mathrm{hon}},j^\star)$ and define the worst honeypot leader payoff $H_{\min}=\min_j\UL(x^{\mathrm{hon}},j)$. Since utilities lie in $[0,1]$, the cost of one honeypot round relative to the SSE ceiling is at most $V_{\sse}-H_{\min}\le1$.

\begin{theorem}[Strict improvement over static SSE]
\label{thm:utility}
Assume targetability and exploitability. Let the deceptive mechanism use a honeypot length $\tau_1$ satisfying the dominance condition and suppose the induced lock-in duration is $\Delta$. If
\[
  \Delta(L^\star-V_{\sse})
  >
  \tau(V_{\sse}-H_{\min}),
\]
then the leader's cumulative utility strictly exceeds $T V_{\sse}$ on the realized path.
\end{theorem}

\begin{proof}
Appendix~\ref{app:proofs} gives the cumulative payoff decomposition.
\end{proof}

Thus if $\tau=O(\sqrt{T\ln T})$ and $\Delta=\Theta(T)$, then $G_{\mathrm{dec}}-T V_{\sse}=\Theta(T)-O(\sqrt{T\ln T})$. Since payoffs are in $[0,1]$, the leader's exploration regret over the signaling phase is at most $\tau$; the formal regret statement is deferred to Appendix~\ref{app:proofs}.

\section{Simulated Experiments}

The experiments are diagnostics for the constructive inequalities rather than proof substitutes. We compare a deceptive leader against an oblivious SSE leader in synthetic non-zero-sum $10\times10$ matrix games and a security-game topology with targets, coverage probabilities, attacker choices, and defender utilities.

\paragraph{Methodology.}
For each seed, compute $V_{\sse}$ by enumerating follower actions and solving the leader's constrained program. The deceptive leader uses $x^{\mathrm{hon}}$ until the target index dominates, then switches to $x^{\mathrm{exp}}$. For direct validation of Theorem~\ref{thm:utility}, we also evaluate the continuation in which the leader returns to an SSE policy after the first escape, matching the proof. We record cumulative utility, target frequency, index gap, switch time, and escape time. An implementation for reproducing the diagnostics is available at \url{https://github.com/suayptalha/optimism-vulnerability}.

\paragraph{Diagnostic results.}
The matrix/security runs use $T=20000$ and $20$ seeds. The certified runs use $T=200000$ and $5$ seeds. Security shows a small positive mean advantage, matrix yields a negative advantage, the non-reverting certified run produces long lock-in but negative total advantage, and the theorem-matched reverting run produces a large positive advantage. A change-point follower collapses lock-in to three rounds, while EXP3 produces short lock-in but still positive average leader advantage through its exploration mixture.

Mean cumulative differences were negative in matrix games and positive in security games. In the certified instance, UCB with reversion achieved $+55178.93$, change-point defense reduced the difference to $-527.04$, and EXP3 averaged $+82764.91$ with high variance. Appendix~\ref{app:experiments} reports the full table. The predicted phase transition is therefore a falsifiable inequality, not a visual story. Around the switch, the relevant empirical object is the index gap
\[
  I_\star(q)-\max_{j\ne j^\star}J_j(q),
\]
which should cross zero near the observed escape time. Follow-up experiments should sweep $c$, stochastic noise, initialization, and defended followers using change-point tests or corruption-robust confidence updates.

\section{Conclusion}

The main lesson is structural: optimism is not merely an exploration heuristic when rewards are generated by strategic opponents. It is a manipulable state variable. Classical SSE is a static equilibrium concept; it cannot price the value of falsified empirical histories because those histories do not exist in the model. In games satisfying explicit targetability and exploitability conditions, a leader can purchase optimism through honeypot rewards and later convert it into utility above the SSE ceiling.

Future work should characterize deception-resistant learning rules for Stackelberg environments: adversarial bandit updates, change-point tests, memory-limited estimators, robust confidence intervals under strategic contamination, and equilibrium concepts in which the leader commits to an information-generation policy. A satisfactory theory must preserve the statistical benefits of exploration without making exploration itself a programmable vulnerability.

\clearpage
\bibliography{deceptive_ucb_refs}

\clearpage
\appendix

\section{Deferred Proofs}
\label{app:proofs}

\subsection{Proof of the Target Mean Lemma}
The empirical mean is the arithmetic average of $N_0$ old samples summing to $N_0\hat\mu_0$ and $\tau_1$ honeypot samples each equal to $\alpha$:
\[
  \hat\mu_{j^\star}(t_0+\tau_1)
  =
  \frac{N_0\hat\mu_0+\tau_1\alpha}{N_0+\tau_1}.
\]
To force this mean above $\alpha-\varepsilon$, require
\[
  \frac{N_0\hat\mu_0+\tau_1\alpha}{N_0+\tau_1}\ge \alpha-\varepsilon.
\]
Multiplying by $N_0+\tau_1$ and canceling $\tau_1\alpha$ gives
\[
  N_0\hat\mu_0\ge N_0\alpha-N_0\varepsilon-\tau_1\varepsilon,
  \qquad
  \tau_1\varepsilon\ge N_0(\alpha-\hat\mu_0-\varepsilon).
\]
Dividing by $\varepsilon>0$ proves the displayed threshold. If $\alpha-\hat\mu_0-\varepsilon<0$, zero additional samples suffice.

\subsection{Proof of the Deterministic Dominance Theorem}
At time $s$, the target index equals
\[
I_{j^\star}(s)=
\frac{N_0\hat\mu_0+\tau_1\alpha}{N_0+\tau_1}
+c\sqrt{\frac{\ln s}{N_0+\tau_1}}.
\]
For every $j\ne j^\star$, no honeypot samples of $j$ occur, hence
\[
I_j(s)=\hat\mu_j(t_0)+c\sqrt{\frac{\ln s}{N_j(t_0)}}
\le
B_\mu+c\sqrt{\ln s}\max_{k\ne j^\star}\frac{1}{\sqrt{N_k(t_0)}}.
\]
Strict dominance is exactly $I_{j^\star}(s)>I_j(s)$ for all $j\ne j^\star$, which is equivalent to the theorem's necessary and sufficient condition. Replacing the maximum by $1$ gives the conservative sufficient condition because $N_j(t_0)\ge1$.

\section{Stochastic Extensions}
\label{app:stochastic}

Let $\mathcal{E}_\delta$ be the event that for all arms $j$, all sample counts $k\le T$, and all adaptive histories generated by the leader policy,
\[
  |\hat\mu_{j,k}-\mu_{j,k}|
  \le
  r_\delta(k)
  :=
  \sigma\sqrt{\frac{2\ln(4nT/\delta)}{k}},
\]
where $\mu_{j,k}$ is the conditional mean average of the $k$ rewards observed on arm $j$. Since rewards are conditionally $\sigma$-sub-Gaussian, a Hoeffding-Azuma bound plus a union bound over $(j,k)$ gives $\mathbb{P}(\mathcal{E}_\delta)\ge1-\delta$.

\begin{theorem}[High-probability stochastic dominance]
\label{thm:stoch-dom}
With probability at least $1-\delta$, the target strictly dominates all non-target actions at switch time $s$ whenever
\[
  \frac{N_0\hat\mu_0+\tau_1\alpha}{N_0+\tau_1}
  -r_\delta(N_0+\tau_1)
  +c\sqrt{\frac{\ln s}{N_0+\tau_1}}
  >
  B_\mu+\max_{j\ne j^\star}r_\delta(N_j(t_0))
  +c\sqrt{\ln s}\max_{j\ne j^\star}N_j(t_0)^{-1/2}.
\]
\end{theorem}

\subsection{Proof of Theorem~\ref{thm:stoch-dom}}
On $\mathcal{E}_\delta$, the post-honeypot target empirical mean is at least
\[
  \frac{N_0\hat\mu_0+\tau_1\alpha}{N_0+\tau_1}
  -r_\delta(N_0+\tau_1),
\]
and each non-target empirical mean is at most $B_\mu+r_\delta(N_j(t_0))$. Adding the exact UCB exploration terms to these lower and upper bounds yields the displayed sufficient inequality. Strict separation implies UCB chooses $j^\star$ at the switch.

\begin{theorem}[High-probability lock-in lower bound]
\label{thm:stoch-lock}
On $\mathcal{E}_\delta$, UCB selects $j^\star$ for at least $Q$ exploitation rounds if, for every $q<Q$,
\[
  \frac{M_\star+q\rho}{N_\star+q}
  -r_\delta(N_\star+q)
  +c\sqrt{\frac{\ln(\tau+q+1)}{N_\star+q}}
  >
  \max_{j\ne j^\star}
  \left[
  \hat\mu_j(t_0)+r_\delta(N_j(t_0))
  +c\sqrt{\frac{\ln(\tau+q+1)}{N_j(t_0)}}
  \right].
\]
\end{theorem}

\subsection{Proof of Theorem~\ref{thm:stoch-lock}}
Condition on $\mathcal{E}_\delta$. If $j^\star$ has been played for $q$ exploitation rounds, its conditional mean average is $(M_\star+q\rho)/(N_\star+q)$ and its empirical mean is no smaller than this value minus $r_\delta(N_\star+q)$. Non-target empirical means remain bounded above by their frozen empirical values plus $r_\delta(N_j(t_0))$. The theorem's inequality makes the target UCB index strictly largest for every $q<Q$, so induction gives at least $Q$ locked rounds.

\subsection{Proof of the Deterministic Lock-in Theorem}
We induct on exploitation rounds. At $q=0$, strict dominance at the switch implies that UCB selects $j^\star$. Suppose the follower has selected $j^\star$ for exactly $q$ exploitation rounds. The target count is $N_\star+q$ and the target cumulative reward is $M_\star+q\rho$. Every non-target statistic remains frozen at its value from $t_0$. Therefore, at the next decision time $\tau+q+1$, UCB selects $j^\star$ if and only if
\[
  \frac{M_\star+q\rho}{N_\star+q}
  +c\sqrt{\frac{\ln(\tau+q+1)}{N_\star+q}}
  >
  \max_{j\ne j^\star}
    \left[
      \hat\mu_j(t_0)+c\sqrt{\frac{\ln(\tau+q+1)}{N_j(t_0)}}
    \right].
\]
The maximal prefix satisfying this inequality is $\Delta$; the first failure is $q_{\mathrm{esc}}$.

\subsection{Proof of Theorem~\ref{thm:utility} and the Sublinear Burden Claim}
During the first $\tau$ rounds, the leader obtains at least $H_{\min}$ each round. During the $\Delta$ locked exploitation rounds, the follower plays $j^\star$, so the leader obtains $L^\star$ each round. If the leader reverts to an SSE action after escape, then
\[
G_{\mathrm{dec}}-T V_{\sse}
\ge
\tau(H_{\min}-V_{\sse})+\Delta(L^\star-V_{\sse}).
\]
The right-hand side is strictly positive precisely when
\[
\Delta(L^\star-V_{\sse})
>
\tau(V_{\sse}-H_{\min}).
\]
If $\Delta=\Theta(T)$ and $\tau=O(\sqrt{T\ln T})$, the positive term is linear because $L^\star>V_{\sse}$, while the signaling loss is $O(\sqrt{T\ln T})$.

\subsection{Exploration Regret Bound}
Define $R_L^{\mathrm{hon}}(T)=\sum_{t=1}^{\tau}(L^\star-\UL(x_t,y_t))$. If $\tau\le K\sqrt{T\ln T}$, then $R_L^{\mathrm{hon}}(T)\le K\sqrt{T\ln T}$.

For every round, $\UL(x_t,y_t)\in[0,1]$ and $L^\star\in[0,1]$, hence
\[
  L^\star-\UL(x_t,y_t)\le 1.
\]
Summing over $\tau$ honeypot and initialization rounds gives
\[
  R_L^{\mathrm{hon}}(T)
  \le
  \sum_{t=1}^{\tau}1
  =
  \tau
  \le
  K\sqrt{T\ln T}.
\]

\section{Additional Mathematical Details}

\begin{lemma}[Lock-in lower bound]
Suppose the frozen competitor index is bounded by
\[
  \max_{j\ne j^\star}J_j(q)\le \Gamma(q),
  \qquad \rho<\Gamma(q)<\alpha,
\]
for $q\le Q$. If
\[
  \frac{M_\star+q\rho}{N_\star+q}>\Gamma(q)
\]
for every $q<Q$, then $\Delta\ge Q$ even without the target exploration bonus.
\end{lemma}

\begin{proof}
For $q<Q$, the target empirical mean alone exceeds every non-target index upper bound. Adding the nonnegative optimism term preserves strict dominance. Thus UCB selects $j^\star$ throughout the first $Q$ exploitation rounds.
\end{proof}

\begin{lemma}[Closed-form sufficient trap length]
If $\Gamma(q)\le\bar\Gamma<\alpha$ for $q\le Q$, then the condition
\[
  Q<
  \frac{M_\star-\bar\Gamma N_\star}{\bar\Gamma-\rho}
\]
implies $\Delta\ge Q$ whenever $\rho<\bar\Gamma$.
\end{lemma}

\begin{proof}
The inequality
\[
  \frac{M_\star+q\rho}{N_\star+q}>\bar\Gamma
\]
is equivalent to
\[
  M_\star-\bar\Gamma N_\star>q(\bar\Gamma-\rho).
\]
The displayed bound makes this true for every $q<Q$.
\end{proof}

\begin{proposition}[Empirical falsification criterion]
For a fixed implementation, define
\[
  \widehat D_T=\frac1S\sum_{s=1}^S
  \left(G_{\mathrm{dec}}^{(s)}(T)-G_{\sse}^{(s)}(T)\right).
\]
If $\widehat D_T<0$, then at least one of the finite-sample conditions $\widehat L^\star>\widehat V_{\sse}$, sufficient lock-in $\widehat\Delta$, or bounded signaling loss fails in the tested generator.
\end{proposition}

\begin{proof}
Theorem~\ref{thm:utility} proves that all three conditions imply $G_{\mathrm{dec}}^{(s)}(T)>G_{\sse}^{(s)}(T)$ pathwise for every seed satisfying them. A negative sample average is therefore only possible if the implication's antecedent fails for at least one seed, or if the implemented baseline differs from the theoretical SSE comparator.
\end{proof}

\begin{proposition}[UCB-parameter misspecification margin]
Suppose the leader designs the switch using $\hat c$ but the follower uses $c$ with $|c-\hat c|\le\kappa$. If the designed target-index margin at time $s$ is larger than
\[
  \kappa\sqrt{\ln s}
  \left(
    \frac{1}{\sqrt{N_{j^\star}(s-1)}}+
    \max_{j\ne j^\star}\frac{1}{\sqrt{N_j(s-1)}}
  \right),
\]
then the actual follower still selects $j^\star$ at time $s$.
\end{proposition}

\begin{proof}
Changing the UCB parameter from $\hat c$ to $c$ perturbs the target bonus by at most $\kappa\sqrt{\ln s/N_{j^\star}(s-1)}$ and any competitor bonus by at most $\kappa\sqrt{\ln s/N_j(s-1)}$. If the designed strict margin exceeds the sum of the worst target decrease and competitor increase, the sign of the margin is preserved.
\end{proof}

\section{Operational Certification}
\label{app:certification}

The assumptions can be checked by solving small linear programs. For each candidate follower response $j$, compute the SSE face value
\[
  V_j=\max_{x\in\Delta(\AL)} \UL(x,j)
  \quad
  \mathrm{s.t.}\quad
  \UF(x,j)\ge \UF(x,k)\ \forall k\in\AF,
\]
and set $V_{\sse}=\max_j V_j$. For each target $j^\star$, the honeypot margin is obtained from
\[
  \max_{x\in\Delta(\AL),\,z} z
  \quad
  \mathrm{s.t.}\quad
  \UF(x,j^\star)-\UF(x,k)\ge z\ \forall k\ne j^\star.
\]
The exploit certificate solves, for each $j^\star$,
\[
  \max_{x\in\Delta(\AL)} \UL(x,j^\star)-V_{\sse}
  \quad
  \mathrm{s.t.}\quad
  \UF(x,k)-\UF(x,j^\star)\ge \gamma_F
  \ \text{for some }k\ne j^\star.
\]
A candidate triple $(j^\star,x^{\mathrm{hon}},x^{\mathrm{exp}})$ is accepted only if the certified lock-in lower bound $Q$ satisfies
\[
  Q(L^\star-V_{\sse})>\tau(V_{\sse}-H_{\min}).
\]
This condition corresponds to a leader that reverts to an SSE policy after the first escape; otherwise post-escape exploit rounds can dominate the accounting and destroy net gain even when lock-in is long.

\paragraph{Certified toy instance.}
The certified instance uses
\[
U_L=
\begin{pmatrix}
0&0\\
1&0
\end{pmatrix},
\qquad
U_F=
\begin{pmatrix}
1&0\\
0.45&0.55
\end{pmatrix},
\]
with $j^\star=1$, $x^{\mathrm{hon}}=(1,0)$, and $x^{\mathrm{exp}}=(0,1)$. Under $x^{\mathrm{exp}}$, the target is follower-suboptimal by $\gamma_F=0.10$, while $L^\star=1$. This instance is intended for controlled sweeps over $c$ and the enforced honeypot length.

\section{Observability and Defenses}
\label{app:partial}

\paragraph{Partial observability.}
The main construction assumes the leader observes $(N_j,\hat\mu_j)_j$. If the leader observes only follower actions and its own selected mixtures, then $N_j$ is still known exactly, while $\hat\mu_j$ can be estimated when the leader knows $\UF$ and the stochastic reward model. A conservative implementation replaces each unknown empirical mean by a confidence interval; the honeypot stops only when the lower confidence bound on the target index exceeds upper confidence bounds on all alternatives. This increases $\tau_1$ by the same concentration radii appearing in Theorem~\ref{thm:stoch-dom}.

\paragraph{Defended followers.}
Change-point tests can flag the trap because the target reward mean changes from $\alpha$ to $\rho$. If a detector requires a detectable drop of order $\epsilon_{\mathrm{det}}$, then any trap with $\alpha-\rho>\epsilon_{\mathrm{det}}$ risks early reset. Corruption-robust UCB or robust Thompson sampling changes the estimator from a plain empirical mean to a robust statistic; in our notation this replaces $r_\delta(k)$ by a larger but attack-aware radius depending on the corruption budget. The deception succeeds only if the honeypot-induced gap exceeds that robust radius, so defenses convert the vulnerability into a quantitative cost increase.

\section{Additional Experiment Details}
\label{app:experiments}

\begin{table}[h]
\centering
\caption{Cumulative reward and trap diagnostics.}
\small
\begin{tabular}{lrrrrrr}
\toprule
Environment & Deceptive & Baseline & Difference & Target freq. & Switch & Escape\\
\midrule
Matrix & $2107.27$ & $9896.70$ & $-7789.43$ & $0.00595$ & $34$ & $36$\\
Security & $3058.05$ & $3048.45$ & $9.60$ & $0.00365$ & $34$ & $36$\\
Cert. UCB, no revert & $4950.00$ & $95008.68$ & $-90058.68$ & $0.02975$ & $1002$ & $4528$\\
Cert. UCB, revert & $150187.61$ & $95008.68$ & $55178.93$ & $0.82936$ & $1002$ & $4528$\\
Cert. CP, revert & $94481.64$ & $95008.68$ & $-527.04$ & $0.52471$ & $1002$ & $1005$\\
Cert. EXP3, revert & $172435.94$ & $89671.03$ & $82764.91$ & $0.95349$ & $1056.2$ & $1070.2$\\
\bottomrule
\end{tabular}
\end{table}

All reported experiments use deterministic follower rewards, UCB exploration parameter $c=2.0$ for the matrix and security diagnostics, and $c=0.2$ with an enforced $1000$-target-sample honeypot in the certified UCB diagnostic. Defended-follower evaluations use the same certified instance with either a simple reward-drop change-point reset or an EXP3-style adversarial-bandit update. The EXP3 row reports means over five seeds and has high variance in both baseline and deceptive rewards. The code release at \url{https://github.com/suayptalha/optimism-vulnerability} contains the simulator, certified instance, and diagnostic logging used for these rows.

\end{document}